\documentclass[11pt,a4paper]{amsart}
\pdfoutput=1
\usepackage{fullpage,amsthm}
\usepackage{amsmath,amsfonts,amssymb} 
\usepackage{hyperref}%[pagebackref]
\usepackage[english]{babel}
\usepackage[capitalise]{cleveref}
\usepackage{slashed}
\usepackage{textcomp}

\numberwithin{equation}{section}

\theoremstyle{plain}
\newtheorem{thm}{Theorem}[section]
\newtheorem{lem}[thm]{Lemma}
\newtheorem{prop}[thm]{Proposition}

\theoremstyle{definition}
\newtheorem{defn}[thm]{Definition}
\newtheorem{example}[thm]{Example}
\newtheorem{remark}[thm]{Remark}

\crefname{thm}{Theorem}{Theorems}
\crefname{lem}{Lemma}{Lemmas}
\crefname{prop}{Proposition}{Propositions}
\crefname{coro}{Corollary}{Corollaries}
\crefname{defn}{Definition}{Definitions}
\crefname{example}{Example}{Examples}
\crefname{remark}{Remark}{Remarks}
\usepackage{xcolor}
\definecolor{MyBlue}{cmyk}{1,0.13,0,0.63} %MidnightBlue: {0.98,0.13,0,0.43}
\definecolor{MyGreen}{cmyk}{0.91,0,0.88,0.52} %ForestGreen: {0.91,0,0.88,0.12}
\definecolor{webbrown}{rgb}{.6,0,0}
\hypersetup{%
  bookmarksnumbered=true,bookmarksopen=false,%true,bookmarksopenlevel=0,%
  plainpages=false,% necessary to prevent duplicate page identifiers
  linktocpage=true,%
  colorlinks=true,breaklinks=true,%
  linkcolor=MyBlue,citecolor=MyGreen,urlcolor=webbrown,%
  pdfpagelayout=OneColumn,%
  pageanchor=true,%
}

\newcommand{\sD}{\slashed{D}}
\newcommand{\I}{\mathbb{I}}
\newcommand{\R}{\mathbb{R}}
\newcommand{\C}{\mathbb{C}}
\newcommand{\Z}{\mathbb{Z}}

\newcommand{\A}{\mathcal{A}}
\newcommand{\mA}{\mathfrak{A}}
\renewcommand{\L}{\mathcal{L}}
\newcommand{\mH}{\mathcal{H}}
\newcommand{\F}{\mathcal{F}}
\newcommand{\lu}{\mathfrak{u}}
\newcommand{\G}{\mathcal{G}}
\newcommand{\Id}{\textnormal{Id}}
\newcommand{\Tr}{\textnormal{Tr}}
\newcommand{\End}{\textnormal{End}}
\newcommand{\Inn}{\textnormal{Inn}}
\newcommand{\Ker}{\textnormal{Ker}}
\renewcommand{\bar}{\overline}
\newcommand{\Sub}[1]{_{\scriptscriptstyle#1}}
\newcommand{\til}[1]{\widetilde{#1}}

\newcommand{\mattwo}[4]{
  \left(\!\!\!\begin{array}{c@{~}c}#1&#2\\#3&#4\\\end{array}\!\!\!\right)
}
\newcommand{\vectwo}[2]{
  \left(\!\!\!\begin{array}{c}#1\\#2\\\end{array}\!\!\!\right)
}
\newcommand{\matfour}[4]{
  \left(\!\!\begin{array}{c@{~}c@{~}c@{~}c}#1\\#2\\#3\\#4\\\end{array}\!\!\right)
}

\DeclareMathOperator{\Ad}{Ad}

\title{Electrodynamics from Noncommutative Geometry}
\author{Koen van den Dungen}
\author{Walter D. van Suijlekom}
\address{IMAPP, Radboud University Nijmegen, Heyendaalseweg 135, 6525 AJ Nijmegen, The Netherlands}
\email{k.vandendungen@student.science.ru.nl; waltervs@math.ru.nl}
\date{March 15, 2011}

\begin{document}

\begin{abstract}
Within the framework of Connes' noncommutative geometry, the notion of an almost commutative manifold can be used to describe field theories on compact Riemannian spin manifolds. The most notable example is the derivation of the Standard Model of high energy physics from a suitably chosen almost commutative manifold. In contrast to such a non-abelian gauge theory, it has long been thought impossible to describe an abelian gauge theory within this framework. The purpose of this paper is to improve on this point. We provide a simple example of a commutative spectral triple based on the two-point space, and show that it yields a $U(1)$ gauge theory. Then, we slightly modify the spectral triple such that we obtain the full classical theory of electrodynamics on a curved background manifold. 
\end{abstract}

\maketitle

%\tableofcontents

\section{Introduction}

The framework of Connes' noncommutative geometry \cite{connes_NCG} provides a generalization of ordinary Riemannian geometry. Within this framework, the notion of an almost commutative manifold (or an almost commutative geometry \cite{Class_IrrACG_I}) can be used to describe gauge field theories on Riemannian spin manifolds. Following a series of articles starting with \cite{CL91,CC96,CC97} this led in \cite{CCM07} to a noncommutative geometrical description of the full Standard Model of high energy physics, including the Higgs mechanism and neutrino mixing. In fact, a non-abelian $SU(N)$-Yang--Mills gauge theory can be described simply by considering matrix-valued functions on a background Riemannian manifold $M$. A key role is played by the adjoint action of the group of unitary matrices on $M$: it acts as $PSU(N)$ for $N$ the rank of the matrices. 

This approach immediately raises a problem if one wishes to describe abelian gauge theories, since $PSU(N)$ is trivial if $N=1$. In fact, it was long believed to be impossible to describe abelian gauge theories within the framework of noncommutative geometry. In this paper, we show that it is very well possible, and we construct a spectral triple ({\it i.e.}\ a noncommutative manifold) that describes a $U(1)$-gauge theory and even the full theory of electrodynamics.

In \cite[Sect.~9.3]{landi} it is claimed that for commutative algebras the gauge fields (and hence the gauge group) are trivial. 
The proof is based on the claim that the left and right action appearing in the adjoint action can be identified for a commutative algebra. Though this claim holds in the case of the canonical triple describing a Riemannian spin manifold, it need not be true for arbitrary commutative algebras. The almost commutative manifold given in \cref{sec:U1} below provides a counter-example.

This paper is organized as follows. We start by reviewing some definitions and results from noncommutative geometry, specializing to the case of almost commutative manifolds. We pay special attention to the form of the gauge group for such manifolds. 
Then, in \cref{sec:U1} we consider the product of spacetime with a two-point space, however, from a noncommutative point of view, tracing back to the early noncommutative models \cite{CL91}.
Essentially, the Riemannian geometry of the product is the usual (commutative) one, but the spin (KO) dimension is different, very similar to \cite{CCM07}.

In \cref{sec:ED} we will show how the above example can be modified to provide a description of one of the simplest examples of a gauge theory in physics, namely electrodynamics. Because of its simplicity, it helps in gaining an understanding of the formulation of gauge theories in terms of almost commutative manifolds, and it provides a first step towards an understanding of the derivation of the Standard Model from noncommutative geometry \cite{CCM07}.

\section{Spectral triples and gauge symmetry}

\subsection{Spectral triples}
\label{sec:triple}

In this section we shall briefly recall the notion of spectral triples. We shall follow the definitions as they appear in \cite[Ch.~1, \S10]{connes-marcolli}, for more details we refer to that book and the references therein. 

A {\it spectral triple} $(\A, \mH, D)$ is given by a unital 
$\ast$-algebra $\A$ represented faithfully as bounded operators on a Hilbert space $\mH$ and a self-adjoint operator $D$ (referred to as a {\it Dirac operator}) with compact resolvent and such that all commutators $[D,a]$ are bounded for $a\in\A$. Note that this implies that the $\A$-module generated by operators of the form $a[D,b]$ ($a,b \in \A$) consists of bounded operators on $\mH$. These differential one-forms will play a key role later, as they will appear as gauge fields. We set accordingly:
\begin{align*}
\Omega^1_D(\A) := \Big\{ \sum_j a_j[D,b_j] \mid a_j, b_j\in\A \Big\} .
\end{align*}

A spectral triple might have additional structure such as a $\Z_2$-grading $\gamma$ on $\mH$, making $\A$ even, and $\Omega^1_D(\A)$ odd. Correspondingly, the Hilbert space decomposes as $\mH = \mH^+ \oplus \mH^-$ into the $\pm 1$ eigenspaces of $\gamma$. In this case, we will call the spectral triple {\it even}, otherwise it is {\it odd}.

Furthermore, an (even) spectral triple has a \emph{real structure} if there is an anti-linear isomorphism $J\colon \mH\rightarrow\mH$ with $J^2 = \varepsilon$, $JD = \varepsilon' DJ$ and, if the spectral triple is even, $J\gamma = \varepsilon'' \gamma J$. The signs $\varepsilon$, $\varepsilon'$ and $\varepsilon''$ determine the \emph{KO-dimension} $n$ modulo $8$ of the spectral triple, according to %the following table.
\begin{align*}
\begin{array}{c|cccccccc}
n & 0 & 1 & 2 & 3 & 4 & 5 & 6 & 7 \\
\hline
\varepsilon & 1 & 1 & -1 & -1 & -1 & -1 & 1 & 1 \\
\varepsilon' & 1 & -1 & 1 & 1 & 1 & -1 & 1 & 1 \\
\varepsilon'' & 1 & & -1 &  & 1 & & -1 & \\
\end{array}
\end{align*}
Moreover, the action of $\A$ is required to satisfy the commutation rule 
\begin{align}
\label{eq:order0}
[a,b^0] = 0 \qquad\forall a,b\in\A ,
\end{align}
where we have defined the right action $b^0$ of $b$ on $\mH$ by
\begin{align*}
b^0 := Jb^*J^{-1} .
\end{align*}
We also require such a commutation relation for $\Omega^1_D(\A)$ with the right action of $\A$, {\it i.e.}
\begin{align}
\big[[D,a],b^0\big] = 0 \qquad\forall a,b\in\A .
\end{align}

\begin{example}
\label{ex:canon}
The motivating example for the definition of spectral triples is the \emph{canonical triple}. Let $M$ be a compact Riemannian spin manifold. We then define the canonical triple by
\begin{align*}
\big(\A, \mH, D\big) = \big(C^\infty(M),L^2(M,S),\sD\big) ,
\end{align*}
where $S$ is the spinor bundle on $M$ and $\sD$ is the canonical Dirac operator given locally by $-i\gamma^\mu\nabla^S_\mu$. Here, $\nabla^S$ is the Levi--Civita connection lifted to the spinor bundle.
Due to the property $ [\sD,a] = -i\, c(\mathrm d a)$, we can identify the differential one-forms $\Omega^1_\sD(C^\infty(M))$ with DeRham differential one-forms (via Clifford multiplication $c$). 

If $M$ is even dimensional (say of dimension $m$), we have a $\Z_2$-grading $\gamma_{m+1}$ 
%We will take $M$ to be $4$-dimensional, and we then have a $\Z_2$-grading $\gamma_5$ 
and an anti-linear isometry $J_M$, which is the charge conjugation operator on the spinors. The Riemannian spin manifold $M$ can be fully described by this canonical triple \cite{connes_foundation,C08}. 
\end{example}

Another special case of a spectral triple is a real even finite spectral triple, given by the data
\begin{align*}
F := \left( \A_F, \mH_F, D_F, \gamma_F, J_F \right) ,
\end{align*}
for a finite dimensional Hilbert space $\mH_F$. The operators $D_F$, $\gamma_F$ and $J_F$, as well as the action of the algebra $\A_F$, are now simply given by matrices acting on $\mH_F$, subject to the aforementioned (anti-)commutation relations. As a first result, we prove
\begin{lem}
\label{lem:class_real}
For any real even finite spectral triple $F$, we can write with respect to the decomposition $\mH = \mH^+ \oplus \mH^-$:
\begin{align*}
\textnormal{KO-dimension }0\colon\quad J_F &= \mattwo{j_+}{0}{0}{j_-} C \quad \textnormal{for symmetric } j_\pm\in U(\mH^\pm) ;\\
\textnormal{KO-dimension }2\colon\quad J_F &= \mattwo{0}{j}{-j^T}{0} C \quad \textnormal{for } jj^* = j^*j = \I ;\\
\textnormal{KO-dimension }4\colon\quad J_F &= \mattwo{j_+}{0}{0}{j_-} C \quad \textnormal{for anti-symmetric } j_\pm\in U(\mH^\pm) ;\\
\textnormal{KO-dimension }6\colon\quad J_F &= \mattwo{0}{j}{j^T}{0} C \quad \textnormal{for } jj^* = j^*j = \I .
\end{align*}
\end{lem}
\begin{proof}
Let the operator $C$ denote complex conjugation. Then any anti-unitary operator $J_F$ can be written as $UC$, where $U$ is some unitary operator on $\mH_F$. We then have $J_F^* = CU^* = U^TC$, and $J_FJ_F^* = UU^* = \I$. 
The different possibilities for the choice of $J_F$ are characterized by the relations $J_F^2 = UCUC = U\bar U = \varepsilon$ and $J_F\gamma_F = \varepsilon''\gamma_FJ_F$. By inserting $\varepsilon,\varepsilon''=\pm1$ according to the KO-dimension, the exact form of $J_F$ can be directly computed by imposing these relations. 
\end{proof}

We now combine the canonical triple for a spin manifold $M$ with the finite spectral triple $F$ to arrive at the noncommutative manifolds that are of particular interest in the context of particle physics.
\begin{defn} 
A \emph{real even almost commutative (spin) manifold} $M \times F$ is described by
\begin{align*}
\left(\A, \mH, D\right) := \left( C^\infty(M,\A_F), L^2(M,S)\otimes \mH_F, \sD\otimes \I + \gamma_{m+1}\otimes D_F \right) ,
\end{align*}
together with a grading $\gamma = \gamma_{m+1}\otimes\gamma_F$ and a real structure $J = J_M \otimes J_F$. 
\end{defn}

\subsection{The gauge group}
\label{sec:gauge_group}

We would like to study the notion of `symmetry' for almost commutative manifolds. The starting point is to define an equivalence of spectral triples. The symmetry is then revealed when it turns out that the bosonic and fermionic action functionals of a spectral triple are identical for equivalent spectral triples (see \cref{sec:act_funct}). We take our definition of equivalent spectral triples from \cite{C96} (cf.\ \cite[\S6.9]{landi}) but make a slight modification by incorporating the algebra isomorphism $\alpha$.

\begin{defn}
Two spectral triples $(\A_1, \mH_1, D_1)$ and $(\A_2, \mH_2, D_2)$, with the associated representations $\pi_j\colon\A_j\rightarrow B(\mH_j)$ for $j=1,2$, are \emph{unitarily equivalent} if there exists a unitary operator $U\colon\mH_1\rightarrow\mH_2$, called the \emph{intertwining operator}, such that 
$$
UD_1U^* = D_2; \qquad U\pi_1(a)U^* = \pi_2(\alpha(a)); \qquad (a \in \A),
$$
where $\alpha$ is an algebra isomorphism $\A_1\rightarrow\A_2$. 

If the two triples are even with grading operators $\gamma_1$ and $\gamma_2$, one also requires that $U\gamma_1U^* = \gamma_2$. If the two triples are real with real structures $J_1$ and $J_2$, one also requires that $UJ_1U^* = J_2$. 
\end{defn}

Note that for a discussion of the equivalence of spectral triples, it is good to explicitly mention the representation of the algebra on the Hilbert space, since the intertwining operator affects this representation. Let us now consider two basic examples of intertwining operators.

\begin{prop}
\label{prop:equiv_triples}
The following two spectral triples are equivalent to $(\A, \mH, D, \gamma, J)$ with representation $\pi\colon\A\rightarrow B(\mH)$:
\begin{enumerate}
\item $(\A, \mH, UDU^*, \gamma, UJU^*)$ with representation $\pi\circ\alpha_u$ for $U=\pi(u)$ with $u\in U(\A)$;
\item $(\A, \mH, UDU^*, \gamma, J)$ with representation $\pi\circ\alpha_u$ for $U=\pi(u)J\pi(u)J^*$ with $u\in U(\A)$,
\end{enumerate}
where $\alpha_u$ is the inner automorphism of $\A$ given by $\alpha_u(a) := uau^*$. 
\end{prop}
\begin{proof}
\begin{enumerate}
\item We only need to check that $U\pi(a)U^* = \pi\circ\alpha_u(a)$ and $U\gamma U^* = \gamma$. The latter relation is evident since the grading operator $\gamma$ commutes with the algebra. We also see that 
\begin{align*}
U\pi(a)U^* = \pi(u)\pi(a)\pi(u)^* = \pi(uau^*) = \pi\circ\alpha_u(a) . 
\end{align*}
\item First, we easily see from \eqref{eq:order0} that $U \equiv \pi(u) \pi(u^*)^0$ is a unitary operator. 
%$UU^* = uJuJ^*(uJuJ^*)^* = uJuJ^*Ju^*J^*u^* = 1$ and similarly $U^*U=1$, so $U$ is indeed a unitary operator. 
The relation $U\gamma U^* = \gamma$ holds since $\pi(u)J\pi(u)J^*\gamma = (\varepsilon'')^2\gamma \pi(u)J\pi(u)J^*$. Since $\pi(u^*)^0$ commutes with $\pi(a)$, we find that 
\begin{align*}
U\pi(a)U^* &= %\pi(u)J\pi(u)J^*\pi(a)J\pi(u)^*J^*\pi(u)^* = \pi(u)\pi(a)J\pi(u)J^*J\pi(u)^*J^*\pi(u)^* \\ &=
\pi(u)\pi(a)\pi(u)^* %= \pi(uau^*) 
= \pi\circ\alpha_u(a) .
\end{align*}
Using the property $\pi(a)^0 J = J \pi(a^*)$ for all $a \in \A$, one can check that $UJU^* = J$.
%\begin{align*}
%UJU^* = %\pi(u)J\pi(u)J^*JJ\pi(u)^*J^*\pi(u)^* = \pi(u)J\pi(u)J\pi(u)^*J^*\pi(u)^* \\
%&= \pi(u)J\pi(u)\pi(u)^*J\pi(u)^*J^* = 
%\pi(u)JJ\pi(u)^*J^* = \varepsilon J^* = J . 
%\end{align*}
\qedhere
\end{enumerate}
\end{proof}

In the first case of \cref{prop:equiv_triples}, the intertwining operator $U$ is given by left multiplication with an element of the unitary group $U(\A)$. In the second case, the action of the operator $U$ on a vector $\xi\in\mH$ can be written as $U\xi = u\xi u^*$, since we identify $JuJ^*$ with the right action of $u^*$. This case is especially interesting because we see that the intertwining operator has no effect on $J$. Thus, the group generated by all operators of the form $U=uJuJ^*$ characterizes all equivalent spectral triples $(\A, \mH, UDU^*, \gamma, J)$, in which only the Dirac operator is affected by the unitary transformation. 
\begin{defn}
\label{defn:gauge_group_NCG}
The \emph{gauge group} $\G(\A)$ of a real spectral triple $(\A, \mH, D, J)$ is defined by
\begin{align*}
\G(\A) := \left\{ U=uJuJ^* \mid u\in U(\A) \right\} .
\end{align*}
\end{defn}

Before we continue to evaluate the exact form of this gauge group, we first consider the following subalgebras of $\A$:
\begin{align}
%\label{eq:subalg1}
\A_J &:= \big\{ a\in\A \mid aJ = Ja \big\} %= \big\{ a\in\A \mid a^0 = a^* \big\}
,\nonumber \\
\label{eq:subalg2}
\til\A_J &:= \big\{ a\in\A \mid aJ = Ja^* \big\} %= \big\{ a\in\A \mid a^0 = a \big\} 
.
\end{align} 
The definition of $\A_J$ is taken from \cite[Prop.~3.3]{CCM07}  (cf.\ \cite[Prop.~1.125]{connes-marcolli}); it is a {\it real} commutative subalgebra in the center of $\A$. We have provided a similar but different definition for $\til\A_J$, since this subalgebra will turn out to be very useful for the description of the gauge group in \cref{prop:gauge_group_NCG}. Note that $aJ = Ja^*$ if and only if $a = a^0$, {\it i.e.} if and only if its left and right action on $\mH$ coincide.

\begin{prop}
\label{prop:subalg}
For a complex algebra $\A$, the subalgebra $\til\A_J$ is an involutive commutative complex subalgebra of the center of $\A$. 
\end{prop}
\begin{proof}
Since we must have $[a,b^0] = 0$ for any $a,b \in\A$, we have $[a,b] = 0$ for any $a\in\A$ and $b\in\til\A_J$, so $\til\A_J$ is contained in the center of $\A$. The requirement $a = a^0$ is complex linear, and also implies that $a^* = (a^0)^* = (a^*)^0$, so we have $a^*\in\til\A_J$ for $a\in\til\A_J$. Finally, we check that for $a,b\in\til\A_J$, we find $(ab)^0 = b^0a^0 = ba = ab$, so $ab\in\til\A_J$. \qedhere
\end{proof}

\begin{prop}
\label{prop:gauge_group_NCG} 
There is a short exact sequence of groups
\begin{align*}
1\rightarrow U(\til\A_J)\rightarrow U(\A)\rightarrow \G(\A)\rightarrow1 ,
\end{align*}
where $\til\A_J$ is defined in \eqref{eq:subalg2}. %In other words, the gauge group $\G(\A)$ is isomorphic to the quotient $U(\A) / U(\til\A_J)$.
\end{prop}
\begin{proof}
The map $\Ad\colon U(\A)\rightarrow \G(\A)$ given by $u\mapsto u (u^*)^0$ is surjective by definition. The commutation relation \eqref{eq:order0} implies that $\Ad$ is a group homomorphism. Its kernel is given by elements $u \in U(\A)$ for which $u (u^*)^0 = 1$. In other words, $\Ker \Ad$ consists of all unitary elements in $\tilde \A_J$.
%which
%, with kernel 
%Moreover, $\Ad$ is a group homomorphism, since the commutation relation $[a,JbJ^*]=0$ implies that $\phi(b)\phi(a) = bJbJ^*aJaJ^* = baJbaJ^* = \phi(ba)$. This map has kernel 
%$\Ker(\Ad) = \{u\in U(\A) \mid uJuJ^*=1\}$. 
%The relation $uJuJ^*=1$ is equivalent to $uJ=Ju^*$, and we note that this 
%is the defining relation of the commutative subalgebra $\til\A_J$ of \eqref{eq:subalg2}. ence we have $\Ker(\phi) = U(\til\A_J)$. 
\end{proof}

From \cref{prop:subalg} we know that $\til\A_J$ is a subalgebra of the center of $\A$. If we denote by $Z$ the subgroup of $U(\A)$ that commutes with $\A$, then the group $U(\til\A_J)$ of the previous proposition is contained in $Z$. The quotient $U(\A) / Z$ yields the group $\Inn(\A)$ of inner automorphisms of the algebra. Proposition \ref{prop:gauge_group_NCG} then implies that in general, the gauge group $\G(\A)$ is larger than $\Inn(\A)$. If $U(\til\A_J)$ is \emph{equal} to $Z$, we have in fact $\Inn(\A) \simeq \G(\A)$.

\subsection{Inner fluctuations and gauge transformations}

In this section we will first define the inner fluctuations of a spectral triple. These inner fluctuations arise from considering Morita equivalences between algebras. For a detailed discussion, we refer to \cite{C96} or \cite[Ch.~1, \S10.8]{connes-marcolli}. In this section, we will simply give the definition. 

Recall the Connes' differential one-forms $\Omega^1_D(\A)$, spanned by operators of the form $a[D,b]$ (with $a,b \in \A$). 
For a real spectral triple (endowed with a real structure $J$) we may replace $D$ by 
\begin{align*}
D_A := D + A + \varepsilon'JAJ^{-1} ,
\end{align*}
for a selfadjoint $A=A^*\in\Omega^1_D$. The elements $A$ are called the \emph{inner fluctuations} of the spectral triple. 

%% Note that in the case $\A = C^\infty(M)$ and $D=\sD$ we have the \emph{commutation relation}
%% \begin{equation}
%% [\sD,a] = -i\; c(da)
%% \end{equation}
%% for all $a\in\A$, so that $\Omega^1_\sD$ is given by the Clifford representation of the $1$-forms $\A^1(M)$. The elements of $\Omega^1_D$ for a general Dirac operator $D$ are therefore regarded as a generalization of $1$-forms. They will be interpreted as gauge potentials or gauge fields. 

In Proposition \ref{prop:equiv_triples} we have seen that an element $U=uJuJ^*\in \G(\A)$ transforms the Dirac operator as $D\rightarrow UDU^*$. Let us now consider the effect of this transformation on the fluctuated Dirac operator $D_A = D + A + \varepsilon'JAJ^*$. A direct calculation shows %\cite[Prop.~5.30]{mythesis} 
that $D_A \mapsto UD_A U^*$ is equivalent to a transformation on $A$ of the form
\begin{align*}
%UD_AU = D_{A^u} , \quad\text{for } 
A^u := uAu^* + u[D,u^*] \in \Omega^1_D .
\end{align*}
In other words, the transformation of a fluctuated Dirac operator can again be written in the form of a fluctuated Dirac operator. This only works because we have restricted $U(\A)$ to the gauge group $\G(\A)$, to make sure that the conjugation operator $J$ remains unchanged. The resulting transformation on the inner fluctuation $A\rightarrow A^u$ shall be interpreted in physics as the gauge transformation of the gauge field.

\subsection{The action functional}
\label{sec:act_funct}
In the previous sections, we have recalled spectral triples and their symmetries. It is now time to introduce interesting invariant functionals on them. 
\begin{defn}[Chamseddine--Connes \cite{CC96}]
\label{defn:spectral_action}
Let $(\A,\mH,D)$ be a spectral triple as above. The {\it spectral action} of a real spectral triple is defined by
\begin{align*}
S_b[A] := \Tr \left(f\Big(\frac{D_A}{\Lambda}\Big)\right) ,
\end{align*}
where $f$ is a positive even function, $\Lambda$ is a cut-off parameter and $D_A$ is the fluctuated Dirac operator.
\end{defn}

The spectral action describes only the action for the (bosonic) gauge fields. For the terms involving fermions and their coupling to the bosons, we need something extra. The precise form of the fermionic action depends on the KO-dimension of the spectral triple. For the purpose of this paper, we will only consider the case of KO-dimension $2$ and give the fermionic action for this case. Referring to the sign table in \cref{sec:triple}, we thus have the relations
\begin{align*}
J^2 &= -1 , & JD &= DJ , & J\gamma &= -\gamma J .
\end{align*}
We use the decomposition $\mH = \mH^+ \oplus \mH^-$ by the grading $\gamma$. Following \cite{CCM07} (cf.\ \cite[Ch.~1, \S16.2-3]{connes-marcolli}) the relations above yield a natural construction of an anti-symmetric bilinear form on $\mH^+$, given for $\xi,\xi'\in\mH^+$ by
\begin{align*}
\mA_D(\xi,\xi') = \langle J\xi,D\xi'\rangle ,
\end{align*}
where $\langle\;,\;\rangle$ is the inner product on $\mH$. 
We define the set of \emph{classical fermions} corresponding to $\mH^+$, 
\begin{align*}
\mH^+_\text{cl} := \{\til\xi \mid \xi\in\mH^+\} ,
\end{align*}
as the set of Grassmann variables $\til\xi$ for $\xi\in\mH^+$. %Assuming that the Hilbert space is separable and has a basis $\{e_j\}$, we can write $\xi = \sum_j \xi_j e_j$. The Grassmann variable $\til\xi$ is then obtained by making every component $\xi_j$ into a Grassmann variable $\til{\xi_j}$, so $\til\xi=\til{\xi_j}e_j$. If the Hilbert space is of the form $L^2(M,E)$ for a vector bundle $E\rightarrow M$ of rank $k$, then there locally exists a frame $\{e_1,\ldots,e_k\}$ such that $\xi(x) = \sum_{j=1}^k \xi_j(x) e_j(x)$. In that case we obtain the Grassmann variable $\til\xi(x) = \sum_{j=1}^k \til{\xi_j}(x) e_j(x)$.

\begin{defn}
\label{defn:fermion_act}
For a real even spectral triple $(\A, \mH, D, \gamma, J)$ of KO-dimension $2$ we define the full \emph{action functional} by
\begin{align*}
S[A,\xi] := S_b[A] + S_f[A,\xi] := \Tr \left(f\Big(\frac{D_A}{\Lambda}\Big)\right) + \frac12 \langle J\til\xi,D_A\til\xi\rangle ,
\end{align*}
for $\til\xi\in\mH^+_\text{cl}$. 
\end{defn}
The factor $1/2$ in front of the \emph{fermionic action} $S_f$ has been chosen for future convenience. 
The action functional $S[A,\tilde\xi]$ is invariant under unitary transformations; in fact, it is invariant under transformations of the gauge group $\G(\A)$. 

Note that we have incorporated two restrictions in the fermionic action $S_f$. The first is that we restrict ourselves to even vectors in $\mH^+$, instead of considering all vectors in $\mH$. The second restriction is that we do not consider the inner product $\langle J\til\xi',D_A\til\xi\rangle$ for two independent vectors $\xi$ and $\xi'$, but instead use the same vector $\xi$ on both sides of the inner product. Each of these restrictions reduces the number of degrees of freedom in the fermionic action by a factor $2$, yielding a factor $4$ in total. It is precisely this approach that solves the problem of fermion doubling pointed out in \cite{lizzi} (see also the discussion in \cite[Ch.~1, \S16.3]{connes-marcolli}).

\subsubsection{The heat expansion}

For future purpose, let us recall some results on heat kernel expansions. For more details we refer the reader to \cite{Gil84}. Suppose we have a vector bundle $E$ on a compact Riemannian manifold $M$, and a second order differential operator $H\colon \Gamma(E)\rightarrow\Gamma(E)$ of the form $H = \Delta^E - Q$, where $\Delta^E$ is the Laplacian of some connection on $E$ and $Q\in\Gamma(\End(E))$. 
For a generalized Laplacian $H$ on $E$ we have the following asymptotic expansion (as $t \to 0$), known as the \emph{heat expansion} \cite[\S1.7]{Gil84}: 
\begin{align*}
\Tr\left(e^{-tH}\right) \sim \sum_{k\geq0} t^{\frac{k-m}2} a_k(H) ,
\end{align*}
where $m$ is the dimension of the manifold, the trace is taken over the Hilbert space $L^2(M,E)$ and the coefficients of the expansion are given by
\begin{align*}
a_k(H) := \int_M a_k(x,H) \sqrt{|g|} d^mx. 
\end{align*}
The coefficients $a_k(x,H)$ are called the Seeley-DeWitt coefficients. We also state here without proof Theorem 4.8.16 from \cite{Gil84}. Note that the conventions used by \cite{Gil84} for the Riemannian curvature $R$ are such that $g^{\mu\rho}g^{\nu\sigma}R_{\mu\nu\rho\sigma}$ is negative for a sphere, in contrast to our own conventions. Therefore we have replaced $s=-R$. Furthermore, we have used that $f_{;\mu}^{\phantom{\mu};\mu} = -\Delta f$ for $f\in C^\infty(M)$. 

\begin{thm}[\cite{Gil84}, Theorem 4.8.16]
\label{thm:seeley-dewitt}
For a generalized Laplacian $H = \Delta^E - Q$ the first three non-zero Seeley-DeWitt coefficients are given by
\begin{gather*}
a_0(x,H) = (4\pi)^{-\frac m2} \Tr(\Id); \qquad a_2(x,H) = (4\pi)^{-\frac m2} \Tr\left(\frac s6 + Q\right) \\ 
a_4(x,H) = (4\pi)^{-\frac m2} \frac1{360} 
\Tr\big(-12 \Delta s 
+ 5 s^2 - 2 R_{\mu\nu} R^{\mu\nu} + 2 R_{\mu\nu\rho\sigma} R^{\mu\nu\rho\sigma} \\
\hspace{5cm} + 60 sQ + 180 Q^2 - 60 \Delta Q + 30 \Omega^E_{\mu\nu} {\Omega^E}^{\mu\nu} \big) ,
\end{gather*}
where the traces are now taken over the fibre $E_x$. Here $s$ is the scalar curvature of the Levi-Civita connection $\nabla$, $\Delta$ is the scalar Laplacian and $\Omega^E$ is the curvature of the connection $\nabla^E$ corresponding to $\Delta^E$. %All $a_k(x,H)$ with odd $k$ vanish.
\end{thm}

Now, assume that the square of the fluctuated Dirac operator $D_A$ is of the form $\Delta^E - Q$ for some vector bundle $E$. Applying the heat expansion on ${D_A}^2$ then yields (as $t \to 0$): 
\begin{equation*}
\Tr\left(e^{-t{D_A}^2}\right) \sim \sum_{k\geq0} t^{\frac{k-m}2} a_k({D_A}^2) ,
\end{equation*}
where the Seeley-DeWitt coefficients are given by \cref{thm:seeley-dewitt}. Then, on writing $f$ as a Laplace transform, we obtain in the case of a $4$-dimensional manifold asymptotically (as $\Lambda\to \infty$):
%\begin{align*}
%f\left(\frac{D_A}{\Lambda}\right) = \int_0^\infty \til g(s) e^{-t{D_A}^2 / \Lambda^2} ds .
%\end{align*} 
%In the case of a $4$-dimensional manifold, 
%for the dominant terms of the heat expansion of the spectral action:
% are given by
\begin{align}
\label{eq:action_expansion}
\Tr \left(f\Big(\frac {D_A}\Lambda\Big)\right) \sim 2 f_4\Lambda^4 a_0({D_A}^2) + 2 f_2\Lambda^2 a_2({D_A}^2) + a_4({D_A}^2) f(0) + O(\Lambda^{-1}),
\end{align}
where $f_j = \int_0^\infty f(v) v^{j-1} dv$ are the momenta of the function $f$ for $j>0$.

\begin{example}
\label{ex:canon_spec_act}
For the canonical triple of a $4$-dimensional spin manifold $M$, we obtain (see \cite[Theorem 1.158]{connes-marcolli}%or \cite[Prop.~6.27]{mythesis}
)
\begin{align*}
\Tr \left(f\Big(\frac \sD\Lambda\Big)\right) \sim \frac1{4\pi^2} \int_M \L_M(g_{\mu\nu}) \sqrt{|g|} d^4x + O(\Lambda^{-1}) ,
\end{align*}
where the gravitational Lagrangian $\L_M$ is given by
\begin{align*}
\L_M(g_{\mu\nu}) := 2f_4\Lambda^4 - \frac16 f_2\Lambda^2 s + f(0) \Big(\frac1{120} \Delta s -\frac1{80} C_{\mu\nu\rho\sigma} C^{\mu\nu\rho\sigma} + \frac{11}{1440}R^*R^* \Big) .
\end{align*}
The first two terms yield the Einstein-Hilbert action including a cosmological constant. In addition, we obtain a higher-order contribution from the Weyl gravity term $C_{\mu\nu\rho\sigma} C^{\mu\nu\rho\sigma}$, as well as a boundary term $\Delta s$ and a topological contribution from $R^*R^*$. 
\end{example}

\section{The two-point space} 
\label{sec:U1}

In this section we will provide a simple example of an almost commutative manifold, based on the product of a spin manifold $M$ with a two-point space $X$. The spectral triple describing this example will have a commutative algebra. As mentioned in the introduction, it has been claimed \cite[Sect.~9.3]{landi} that the inner fluctuation $A + JAJ^*$ vanishes for commutative algebras. The proof is based on the claim that the left and right action can be identified, i.e.\ $a=a^0$, for a commutative algebra. Though this claim holds in the case of the canonical triple describing a spin manifold, it need not be true for arbitrary commutative algebras. The spectral triple given in this section provides a counter-example. 

What can be said for a commutative algebra, is that there exist no non-trivial inner automorphisms. It is thus an important insight here that the gauge group $\G(\A)$, as defined in \cref{defn:gauge_group_NCG}, is larger than the group of inner automorphisms, so that a commutative spectral triple may still lead to a non-trivial gauge group. In fact, we will show that our example given below describes an abelian $U(1)$ gauge theory. 

\subsection{A two-point space}
\label{sec:two-point}

We shall consider a finite spectral triple $F\Sub{X}$ corresponding to the two-point space $X = \{x,y\}$. A complex function on this space is simply determined by two complex numbers. The algebra of functions on $X$ is then given by $C(X) = \C^2$. Let us consider the \emph{even} finite spectral triple $F\Sub{X}$ given by
\begin{align*}
\left( C(X), \mH_F, D_F, \gamma_F \right) .
\end{align*}
We require that the representation $C(X)\rightarrow B(\mH_F)$ is faithful, which implies that the Hilbert space $\mH_F$ must be at least $2$-dimensional. Thus, the simplest possible choice is to take $\mH_F = \C^2$. We use the $\Z_2$-grading $\gamma_F$ to decompose $\mH_F = \mH_F^+ \oplus \mH_F^- = \C\oplus\C$ into the two eigenspaces $\mH_F^\pm = \{\psi\in\mH_F \mid \gamma_F\psi = \pm\psi \}$. Hence, we can write 
\begin{align*}
\gamma_F = \mattwo{1}{0}{0}{-1} .
\end{align*}
Because of the relations $[\gamma_F,a]=0$ and $D_F\gamma_F=-\gamma_FD_F$, the self-adjoint Dirac operator must be off-diagonal and the action of an element $a\in\A_F$ on $\psi\in\mH_F$ can be written as 
\begin{align}
\label{eq:rep_U1}
a\psi &= \mattwo{a_+}{0}{0}{a_-} \vectwo{\psi_+}{\psi_-} .
\end{align}
Thus, the even finite spectral triple $F\Sub{X}$ we will study in this section is given by
\begin{align}
\label{eq:triple_U1}
\left(\A_F, \mH_F, D_F, \gamma_F \right) = \left( \C^2, \C^2, \mattwo{0}{t}{\bar t}{0}, \mattwo{1}{0}{0}{-1} \right) ,
\end{align}
for some complex parameter $t\in\C$, and with the representation of $\A_F$ on $\mH_F$ given by \eqref{eq:rep_U1}.

\begin{prop}
\label{prop:zero_D_F}
The even finite spectral triple $F\Sub{X}$ given by \eqref{eq:triple_U1} can only have a real structure if $D_F = 0$. 
\end{prop}
\proof
We must have ${J_F}^2 = \varepsilon$ and $J_F\gamma_F = \varepsilon''\gamma_FJ_F$, and we shall consider all possible (even) KO-dimensions separately. Thus, we apply Lemma \ref{lem:class_real} to the finite spectral triple $F\Sub{X}$ given above and, for each even KO-dimension, also impose the relations $[a,b^0]=0$ and $\big[[D_F,a],b^0\big]=0$. This gives:
\begin{description}
\item[{\it KO-dimension $0$}] \mbox{}\\
We have $J_F = \mattwo{j_+}{0}{0}{j_-} C$ for $j_\pm\in U(1)$. For $b=\mattwo{b_+}{0}{0}{b_-}$ we then obtain 
$$
b^0 = \mattwo{j_+b_+\bar{j_+}}{0}{0}{j_-b_-\bar{j_-}} = b ,
$$
and see that this indeed commutes with the left action of $a\in\C^2$. Next, we check the order one condition
\begin{align*}
0 = \big[[D_F,a],b^0\big] &= %\left[\mattwo{0}{ta_--a_+t}{\bar t a_+-a_-\bar t}{0},\mattwo{b_+}{0}{0}{b_-}\right] =
 (a_+-a_-)(b_+-b_-)D_F .
\end{align*}
Since this must hold for all $a,b\in\C^2$, we conclude that we must require $D_F = 0$. 

\item[{\it KO-dimension $2$}] \mbox{}\\
We have $J_F = \mattwo{0}{j}{-j}{0} C$ for $j\in U(1)$. We now obtain 
$$
b^0 = \mattwo{jb_-\bar{j}}{0}{0}{jb_+\bar{j}} = \mattwo{b_-}{0}{0}{b_+} ,
$$
and see that this indeed commutes with the left action of $a\in\C^2$. Next, we check the order one condition
\begin{align*}
0 = \big[[D_F,a],b^0\big] &= %\left[\mattwo{0}{ta_--a_+t}{\bar t a_+-a_-\bar t}{0},\mattwo{b_-}{0}{0}{b_+}\right] = 
(a_+-a_-)(b_--b_+)D_F .
\end{align*}
Again we conclude that we must require $D_F = 0$. 

\item[{\it KO-dimension $4$}] \mbox{}\\
We have $J_F$ %= \mattwo{j_+}{0}{0}{j_-} C$ 
of the same form as in KO-dimension 0, but now with $j_\pm = -j_\pm^T\in U(1)$. This implies that $j_\pm = 0$, so the given finite spectral triple cannot have a real structure in KO-dimension $4$.

\item[{\it KO-dimension $6$}] \mbox{}\\
We have $J_F = \mattwo{0}{j}{j}{0} C$ for $j\in U(1)$. We again obtain
$$
b^0 = \mattwo{jb_-\bar{j}}{0}{0}{jb_+\bar{j}} = \mattwo{b_-}{0}{0}{b_+} ,
$$
just as for KO-dimension $2$. Hence again the commutation rules are only satisfied for $D_F = 0$. \qedhere
\end{description}
\endproof

\subsection{The almost commutative manifold}
\label{sec:product_space_U1}

Let us now consider the product of the finite spectral triple $F\Sub{X}$ of the two-point space, as described by \eqref{eq:triple_U1}, with the canonical triple describing a compact Riemannian spin manifold $M$, as in \cref{ex:canon}. From here on we will take $M$ to be $4$-dimensional. Thus we consider the almost commutative manifold $M\times F\Sub{X}$ given by the data
\begin{align*}
M\times F\Sub{X} := \Big( C^\infty(M,\C^2), L^2(M,S)\otimes\C^2, \sD\otimes\I, \gamma_5\otimes\gamma_F, J_M\otimes J_F \Big) ,
\end{align*}
where we still need to make a choice for $J_F$. The algebra of this almost commutative manifold is given by $C^\infty(M,\C^2) \simeq C^\infty(M)\oplus C^\infty(M) \simeq C^\infty(M \times X)$. Thus, the underlying space $N := M\times X \simeq M \sqcup M$ consists of the disjoint union of two identical copies of the space $M$, and we can write $C^\infty(N) = C^\infty(M)\oplus C^\infty(M)$. We can also decompose the total Hilbert space as $\mH = L^2(M,S) \oplus L^2(M,S)$. For $a,b\in C^\infty(M)$ and $\psi,\phi\in L^2(M,S)$, an element $(a,b)\in C^\infty(N)$ then simply acts on $(\psi,\phi)\in\mH$ as $(a,b) (\psi,\phi) = (a\psi,b\phi)$. 

\subsubsection{Distances}
To any spectral triple $(\A,\mH,D)$ one can associate a distance function on the space of states on $\A$: 
$$
d_D(p,q) = \sup \left\{ |p(a)- q(a) | \colon a\in\A, \|[D,a]\|\leq1 \right\} ,
$$
For the canonical triple, $\A = C^\infty(M)$ whose state space is homeomorphic to $M$. It turns out that $d_\sD$ is equal to the geodesic distance $d_g$ between points $p$ and $q$ on $M$.

We will use this formula as a generalized notion of distance, so on our finite spectral triple $F\Sub{X}$ we can write
\begin{align*}
d_{D_F}(x,y) = \sup \left\{ |a(x) - a(y)| \colon a\in\A_F, \|[D_F,a]\|\leq1 \right\} .
\end{align*}
Note that we now have only two distinct points $x$ and $y$ in the space $X$, and we shall calculate the distance between these points (cf.\ \cite[Sect.~6.8]{landi}). An element $a\in\C^2=C(X)$ is specified by two complex numbers $a(x)$ and $a(y)$ %given by $\mattwo{a(x)}{0}{0}{a(y)}$, so 
and its commutator with $D_F$ becomes
\begin{align*}
[D_F,a] =% \mattwo{0}{t}{\bar t}{0}\mattwo{a(x)}{0}{0}{a(y)} - \mattwo{a(x)}{0}{0}{a(y)}\mattwo{0}{t}{\bar t}{0} = 
\big(a(y)-a(x)\big) \mattwo{0}{t}{-\bar t}{0} .
\end{align*}
The norm of this commutator is $|a(y)-a(x)|\,|t|$, so $\|[D_F,a]\| \leq 1$ if and only if $|a(y)-a(x)| \leq \frac{1}{|t|}$. We thus obtain that the distance between the two points $x$ and $y$ is given by
\begin{align*}
d_{D_F}(x,y) = \frac{1}{|t|} .
\end{align*}
If there is a real structure $J_F$, we have $t=0$ by \cref{prop:zero_D_F}, so then the distance between the two points becomes infinite. 

Let $p$ be a point in $M$, and write $(p,x)$ or $(p,y)$ for the two corresponding points in $N=M\times X$. A function $a\in C^\infty(N)$ is then determined by two functions $a_x,a_y\in C^\infty(M)$, given by $a_x(p) := a(p,x)$ and $a_y(p) := a(p,y)$. Now consider the distance function on $N$ given by 
\begin{align*}
d_{\sD\otimes\I}(n_1,n_2) = \sup \left\{ |a(n_1) - a(n_2)| \colon a\in\A, \|[\sD\otimes\I,a]\|\leq1 \right\} .
\end{align*}
If $n_1$ and $n_2$ are points in the same copy of $M$, for instance if $n_1=(p,x)$ and $n_2=(q,x)$ for points $p,q\in M$, then their distance is determined by $|a_x(p) - a_x(q)|$, for functions $a_x\in C^\infty(M)$ for which $\|[\sD,a_x]\|\leq1$. Thus, in this case we obtain that we recover the geodesic distance on $M$, {\it i.e.}\ $d_{\sD\otimes\I}(n_1,n_2) = d_g(p,q)$. 

However, if $n_1$ and $n_2$ are points in a different copy of $M$, for instance if $n_1=(p,x)$ and $n_2=(q,y)$, then their distance is determined by $|a_x(p) - a_y(q)|$ for two functions $a_x,a_y\in C^\infty(M)$, such that $\|[\sD,a_x]\|\leq1$ and $\|[\sD,a_y]\|\leq1$. These latter requirements however yield no restriction on $|a_x(p) - a_y(q)|$, so in this case the distance between $n_1$ and $n_2$ is infinite. We thus find that the space $N$ is given by two disjoint copies of the Riemannian manifold $M$, which are separated by an infinite distance.

\subsection{\texorpdfstring{$U(1)$}{U(1)} gauge theory}
\label{sec:U1_gauge}

We would now like to derive the gauge theory that corresponds to the almost commutative manifold $M\times F\Sub{X}$. Recall that the gauge group $\G(\A)$ is given by the quotient $U(\A) / U(\til\A_J)$, so if we wish to obtain a nontrivial gauge group, we need to choose $J$ such that $U(\til\A_J) \neq U(\A)$. By looking at the form of $J_F$ for the different (even) KO-dimensions, as given in \cref{sec:two-point}, we conclude that $F\Sub{X}$ must have KO-dimension $2$ or $6$. In analogy with the noncommutative description of the Standard Model \cite{CCM07} we choose to work in KO-dimension 6. The almost commutative manifold $M\times F\Sub{X}$ then has KO-dimension $6+4\mod 8 = 2$. This means that we can use \cref{defn:fermion_act} to calculate the fermionic action. Therefore, we will consider the finite spectral triple $F\Sub{X}$ given by the data 
\begin{align*}
F\Sub{X} := \left(\A_F, \mH_F, D_F, \gamma_F, J_F \right) := \left( \C^2, \C^2, 0, \mattwo{1}{0}{0}{-1}, \mattwo{0}{C}{C}{0} \right) ,
\end{align*}
which define a real even finite spectral triple of KO-dimension $6$. Now, let us derive the gauge group.

\begin{prop}
\label{prop:gauge_group_U1}
The gauge group $\G(\A_F)$ of the two-point space is given by $U(1)$. %$\simeq \left\{\vectwo{\lambda}{\lambda^{-1}} \mid \lambda\in U(1)\right\}$
\end{prop}
\begin{proof}
First, note that $U(\A_F) = U(1) \times U(1)$. We will show that $U( (\til\A_F)_{J_F} ) \equiv U(\A_F) \cap (\til\A_F)_{J_F} \simeq U(1)$ so that the quotient $\G(\A_F) \simeq U(1)$ as claimed. Indeed, for $a \in \C^2$ to be in $(\til \A_F)_{J_F}$ it has to satisfy $J_F a^* J_F =a $. Since 
$$
 J_Fa^*J_F^* = \mattwo{0}{C}{C}{0} \mattwo{\bar a_1}{0}{0}{\bar a_2} \mattwo{0}{C}{C}{0} = \mattwo{a_2}{0}{0}{a_1} ,
$$
this is the case if and only if $a_1 = a_2$. Thus, $(\til \A_F)_{J_F} \simeq \C$ whose unitary elements form the group $U(1)$, contained in $U(\A_F)$ as the diagonal subgroup.
%% First, for $u = \mattwo{u_1}{0}{0}{u_2} \in U(\C^2)$ we have $uu^* = u^*u = \I$, which implies $u_1\bar{u_1} = u_2\bar{u_2} = 1$. This means that $u_1,u_2\in U(1)$, and we have $U(\C^2) = U(1)\oplus U(1)$. Second, the subgroup $H_F = U(\A_F) \cap (\til\A_F)_{J_F}$ is determined by the condition that $J_Fu^*J_F^* = \mattwo{0}{C}{C}{0} \mattwo{\bar u_1}{0}{0}{\bar u_2} \mattwo{0}{C}{C}{0} = \mattwo{u_2}{0}{0}{u_1}$ is equal to $u = \mattwo{u_1}{0}{0}{u_2}$. This implies that $u_1=u_2$, so we obtain the subgroup
%% \begin{align*}
%% H_F = U(\A_F) \cap (\til\A_F)_{J_F} = \left\{ \vectwo{\lambda}{\lambda} \mid \lambda\in U(1) \right\} \simeq U(1) .
%% \end{align*}
%% By \cref{defn:gauge_group_NCG}, we then obtain $\G(\A_F) := U(\A_F) / H_F$. Let us construct the short exact sequence 
%% \begin{align*}
%% 1 \rightarrow H_F \rightarrow U(\A_F) \rightarrow U(1) \rightarrow 1 ,
%% \end{align*}
%% by defining the group homomorphism $\varphi\colon U(\A_F)\rightarrow U(1)$ as $\varphi(u) := u_1\bar{u_2}$. We immediately see that $\varphi(u) = 1$ if and only if $u_1=u_2$, so the kernel of $\varphi$ is indeed the subgroup $H_F$. We thus see that the gauge group $\G(\A_F)$ is isomorphic to $U(1)$. 
\end{proof}

We will now derive the gauge field for the almost commutative manifold $M\times F\Sub{X}$. Thus, we need to calculate the inner fluctuations of the Dirac operator. For $a,b\in C^\infty(M,\C^2)$, an inner fluctuation $A$ takes the form
\begin{align*}
A = a[D,b] = a[\sD\otimes\I,b] = i\gamma^\mu \otimes a\partial_\mu b =: \gamma^\mu\otimes A_\mu ,
\end{align*}
where we have defined the hermitian field $A_\mu \in C^\infty(M,\R^2)$. Using the relation $J_M\gamma^\mu = -\gamma^\mu J_M$, the total inner fluctuation is then given by
\begin{align}
\label{eq:inn_fluc}
A + JAJ^* = \gamma^\mu \otimes (A_\mu - J_FA_\mu J_F^*) =: \gamma^\mu \otimes B_\mu .
\end{align}
An arbitrary hermitian field of the form $A_\mu=ia\partial_\mu b$ would be given by $\mattwo{X^1_\mu}{0}{0}{X^2_\mu}$, for two $U(1)$ gauge fields $X^1_\mu,X^2_\mu\in C^\infty(M,\R)$. However, $A_\mu$ only appears in the combination 
\begin{align*}
B_\mu = A_\mu - J_FA_\mu J_F^{-1} = \mattwo{X^1_\mu}{0}{0}{X^2_\mu} - \mattwo{X^2_\mu}{0}{0}{X^1_\mu} =: \mattwo{Y_\mu}{0}{0}{-Y_\mu} = Y_\mu\otimes\gamma_F ,
\end{align*}
where we have defined the $U(1)$ gauge field $Y_\mu := X^1_\mu-X^2_\mu \in C^\infty(M,\R) = C^\infty(M,i\,\lu(1))$. Thus, the fact that we only have the combination $A+JAJ^*$ effectively identifies the $U(1)$ gauge fields on the two copies of $M$, so that $A_\mu$ is determined by only one $U(1)$ gauge field. We summarize:
\begin{prop}
\label{prop:gauge-field}
The inner fluctuations of the almost commutative manifold $M \times F_X$ described above are parametrized by a $U(1)$-gauge field $Y_\mu$ as
$$
D \mapsto D' = D + \gamma^\mu Y_\mu \otimes \gamma_F.
$$
The action of the gauge group $\G(\A) \simeq C^\infty(M, U(1))$ on $D'$ by conjugation is implemented by
$$
Y_\mu \mapsto Y_\mu - i u \partial_\mu u^*, \qquad (u \in \G(\A)). 
$$ 
\end{prop}

So far we have seen that the almost commutative manifold $M\times F\Sub{X}$ describes a gauge theory with local gauge group $U(1)$, where the inner fluctuations of the Dirac operator provide the $U(1)$ gauge field $Y_\mu$. The question arises whether this model is suitable for a description of (classical) electrodynamics. There appear to be two problems, even before considering the fermionic action $S_f$ explicitly. First, by \cref{prop:zero_D_F}, the finite Dirac operator $D_F$ must vanish. However, we want our fermions to be massive, and for this purpose we need a finite Dirac operator that is non-zero. 

Second, from \cite[Ch.7, \S5.2]{coleman}, we find the Euclidean action for a free Dirac field:
\begin{align}
\label{eq:coleman}
S = - \int i \bar\psi (\gamma^\mu\partial_\mu - m) \psi d^4x ,
\end{align}
where the fields $\psi$ and $\bar\psi$ must be considered as \emph{totally independent variables}. Thus, we require that the fermionic action $S_f$ should also yield two \emph{independent} Dirac spinors. Let us write $\left\{e, \bar e\right\}$ for the set of orthonormal basis vectors of $\mH_F$, where $e$ is the basis element of $\mH_F^+$ and $\bar e$ of $\mH_F^-$. Note that on this basis, we have $J_Fe=\bar e$, $J_F\bar e=e$, $\gamma_Fe=e$ and $\gamma_F\bar e=-\bar e$. The total Hilbert space $\mH$ is given by $L^2(M,S)\otimes\mH_F$. Since we can also decompose $L^2(M,S) = L^2(M,S)^+ \oplus L^2(M,S)^-$ by means of $\gamma_5$, we obtain that the positive eigenspace $\mH^+$ of $\gamma=\gamma_5\otimes\gamma_F$ is given by
\begin{align*}
\mH^+ = L^2(M,S)^+\otimes\mH_F^+ \oplus L^2(M,S)^-\otimes\mH_F^- .
\end{align*}
An arbitrary vector $\xi\in\mH^+$ can then uniquely be written as 
\begin{align*}
\xi =  \psi_L \otimes e + \psi_R \otimes \bar e ,
\end{align*}
for two Weyl spinors $\psi_L\in L^2(M,S)^+$ and $\psi_R\in L^2(M,S)^-$. One should note here that this vector $\xi$ is completely determined by only one Dirac spinor $\psi := \psi_L + \psi_R$, instead of the required two independent spinors. Thus, the restrictions that are incorporated into the fermionic action of \cref{defn:fermion_act} are such that the present example is in fact too restricted.

\section{Electrodynamics}
\label{sec:ED}

\subsection{The finite spectral triple}

Inspired by the previous section, which shows that one can use the framework of noncommutative geometry to describe a gauge theory with the abelian gauge group $U(1)$, we shall now attempt to describe the full theory of electrodynamics. There are two changes we need to make to the $U(1)$ gauge theory of the previous section. We need to incorporate a non-zero finite Dirac operator $D_F$ to obtain mass terms for the fermions, and we also need to obtain two independent Dirac spinors in the fermionic action. Both these changes can be simply obtained by doubling our finite Hilbert space.

We start with the same algebra $C^\infty(M,\C^2)$ that corresponds to the space $N = M\times X \simeq M\sqcup M$. The finite Hilbert space will now be used to describe four particles, namely both the left-handed and the right-handed electrons and positrons. We will choose the orthonormal basis $\left\{e_L, e_R, \bar{e_L}, \bar{e_R}\right\}$ for $\mH_F = \C^4$, with respect to the standard inner product. 
%Hence, for two vectors $h_F = a e_L + b e_R + c \bar{e_L} + d \bar{e_R}$ and $h_F' = a' e_L + b' e_R + c' \bar{e_L} + d' \bar{e_R}$ in $\mH_F$, we take the inner product on $\mH_F$ to be given by
%\begin{align*}
%\langle h_F, h_F'\rangle := \bar{a}a' + bb' + cc' + dd' .
%\end{align*}
%
The subscript $L$ denotes left-handed particles, and the subscript $R$ denotes right-handed particles, and we take $\gamma_F e_L = e_L$ and $\gamma_F e_R = -e_R$. 

We will choose $J_F$ such that it interchanges particles with their anti-particles, so $J_Fe_R = \bar{e_R}$ and $J_Fe_L = \bar{e_L}$. As in \cref{sec:U1_gauge}, we will choose the real structure such that is has KO-dimension $6$, so we have $J_F^2 = \I$ and $J_F\gamma_F = -\gamma_FJ_F$. This last relation implies that the element $\bar{e_R}$ is left-handed and $\bar{e_L}$ is right-handed. Hence, the grading $\gamma_F$ and the conjugation operator $J_F$ are given by
\begin{align*}
\gamma_F &= \matfour{1&0&0&0}{0&-1&0&0}{0&0&-1&0}{0&0&0&1} , & J_F &= \matfour{0&0&C&0}{0&0&0&C}{C&0&0&0}{0&C&0&0} .
\end{align*}

The grading $\gamma_F$ decomposes the Hilbert space $\mH_F$ into $\mH_L\oplus\mH_R$, where the bases of $\mH_L$ and $\mH_R$ are given by $\{e_L,\bar{e_R}\}$ and $\{e_R,\bar{e_L}\}$, respectively. We can also decompose the Hilbert space into $\mH_e\oplus\mH_{\bar e}$, where $\mH_e$ contains the electrons $\{e_L,e_R\}$, and $\mH_{\bar e}$ contains the positrons $\{\bar{e_L},\bar{e_R}\}$. 

The elements $a\in\A_F=\C^2$ are now represented on the basis $\left\{e_L, e_R, \bar{e_L}, \bar{e_R}\right\}$ as
\begin{align}
\label{eq:rep_ED}
a = \vectwo{a_1}{a_2} \mapsto \matfour{a_1&0&0&0}{0&a_1&0&0}{0&0&a_2&0}{0&0&0&a_2} .
\end{align}
Note that this representation commutes with the grading, as it should. We can also easily check that $[a,b^0] = 0$ for $b^0 := J_Fb^*J_F^*$, since both the left and the right action are given by diagonal matrices. For now, we will still take $D_F = 0$, and hence the order one condition is trivially satisfied. We have now obtained the following result:
\begin{prop}
The real even finite spectral triple 
\begin{align*}
F\Sub{ED} := (\C^2, \C^4, 0, \gamma_F, J_F) 
\end{align*}
as given above defines a real even finite spectral triple of KO-dimension $6$. 
\end{prop}

\subsubsection{A non-trivial finite Dirac operator}

Let us now consider the possibilities for adding a non-zero Dirac operator to the finite spectral triple $F\Sub{ED}$. Since $D_F\gamma_F = -\gamma_FD_F$, the Dirac operator obtains the form 
\begin{align*}
D_F = \matfour{0&d_1&d_2&0}{\bar d_1&0&0&d_3}{\bar d_2&0&0&d_4}{0&\bar d_3&\bar d_4&0} .
\end{align*}
Next, we impose the commutation relation $D_FJ_F = J_FD_F$, which implies $d_1 = \bar d_4$.
%% \begin{align*}
%% \matfour{d_2C&0&0&d_1C}{0&d_3C&\bar d_1C&0}{0&d_4C&\bar d_2C&0}{\bar d_4C&0&0&\bar d_3C} = \matfour{C\bar d_2&0&0&Cd_4}{0&C\bar d_3&C\bar d_4&0}{0&Cd_1&Cd_2&0}{C\bar d_1&0&0&Cd_3} .
%% \end{align*}
%% This imposes the relation $d_1 = \bar d_4$. 
For the order one condition, we calculate
\begin{align*}
[D_F,a] &= %\matfour{0&d_1a_1-a_1d_1&d_2a_2-a_1d_2&0}{\bar d_1a_1-a_1\bar d_1&0&0&d_3a_2-a_1d_3}{\bar d_2a_1-a_2\bar d_2&0&0&\bar d_1a_2-a_2\bar d_1}{0&\bar d_3a_1-a_2\bar d_3&d_1a_2-a_2d_1&0} = 
(a_1-a_2) \matfour{0&0&-d_2&0}{0&0&0&-d_3}{\bar d_2&0&0&0}{0&\bar d_3&0&0} .
\end{align*}
which then imposes the condition
\begin{align*}
0 = \big[[D_F,a],b^0\big] %= (a_1-a_2) \left[ \matfour{0&0&-d_2&0}{0&0&0&-d_3}{\bar d_2&0&0&0}{0&\bar d_3&0&0} , \matfour{b_2&0&0&0}{0&b_2&0&0}{0&0&b_1&0}{0&0&0&b_1} \right] \\
&= (a_1-a_2)(b_2-b_1) \matfour{0&0&d_2&0}{0&0&0&d_3}{\bar d_2&0&0&0}{0&\bar d_3&0&0} . 
\end{align*}
Since this must hold for all $a,b\in\C^2$, we must require that $d_2=d_3=0$. To conclude, the Dirac operator only depends on one complex parameter and is given by
\begin{align}
\label{eq:Dirac}
D_F = \matfour{0&d&0&0}{\bar d&0&0&0}{0&0&0&\bar d}{0&0&d&0} .
\end{align}
From here on, we will consider the finite spectral triple $F\Sub{ED}$ given by
\begin{align*}
F\Sub{ED} := (\C^2, \C^4, D_F, \gamma_F, J_F) .
\end{align*}

\subsection{The almost commutative manifold}

By taking the product with the canonical triple, our almost commutative manifold (of KO-dimension $2$) under consideration is given by
\begin{align*}
M\times F\Sub{ED} := \left( C^\infty(M,\C^2), L^2(M,S)\otimes\C^4, \sD\otimes\I + \gamma_5\otimes D_F, \gamma_5\otimes\gamma_F, J_M\otimes J_F \right) .
\end{align*}
As in \cref{sec:U1}, the algebra decomposes as $C^\infty(M,\C^2) = C^\infty(M)\oplus C^\infty(M)$, and we now decompose the Hilbert space as $\mH = (L^2(M,S)\otimes\mH_e)\oplus(L^2(M,S)\otimes\mH_{\bar e})$. The action of the algebra on $\mH$, given by \eqref{eq:rep_ED}, is then such that one component of the algebra acts on the electron fields $L^2(M,S)\otimes\mH_e$, and the other component acts on the positron fields $L^2(M,S)\otimes\mH_{\bar e}$. 

The derivation of the gauge group for $F\Sub{ED}$ is exactly the same as in \cref{prop:gauge_group_U1}, so again we have the finite gauge group $\G(\A_F) \simeq U(1)$. The field $B_\mu := A_\mu - J_FA_\mu J_F^*$ now takes the form
\begin{align}
\label{eq:gauge_field}
B_\mu = \matfour{Y_\mu&0&0&0}{0&Y_\mu&0&0}{0&0&-Y_\mu&0}{0&0&0&-Y_\mu} \qquad\text{for } Y_\mu(x) \in \R .
\end{align}
Thus, we again obtain a single $U(1)$ gauge field $Y_\mu$, carrying an action of the gauge group $\G(\A) \simeq C^\infty(M, U(1))$ (as in Proposition \ref{prop:gauge-field}). 

As mentioned before, our space $N$ consists of two copies of $M$, and the distance between these two copies is infinite (cf.\ \cref{sec:product_space_U1}). Now, we have introduced a non-zero Dirac operator, but it commutes with the algebra, i.e.\ $[D_F,a]=0$ for all $a\in\A$. Therefore, the distance between the two copies of $M$ is still infinite. 

To summarize, the $U(1)$ gauge theory arises from the geometric space $N=M \sqcup M$ as follows. On one copy of $M$, we have the vector bundle $S\otimes(M\times\mH_e)$, and on the other copy the vector bundle $S\otimes(M\times\mH_{\bar e})$. The gauge fields on each copy of $M$ are effectively identified to each other. The electrons $e$ and positrons $\bar e$ are then both coupled to the same gauge field, and as such the gauge field provides an interaction between electrons and positrons.

\subsection{The Lagrangian}

We are now ready to explicitly calculate the Lagrangian that corresponds to the almost commutative manifold $M\times F\Sub{ED}$, and we will show that this yields the usual Lagrangian for electrodynamics (on a curved background manifold), as well as a purely gravitational Lagrangian. The action functional for an almost commutative manifold, as defined in \cref{defn:fermion_act}, consists of the spectral action $S_b$ and the fermionic action $S_f$, which we will calculate separately. 

\subsubsection{The spectral action}

Before we can calculate the spectral action, we first need to study the fluctuated Dirac operator in a little more detail. As in \eqref{eq:inn_fluc}, we have $A+JAJ^* = \gamma^\mu\otimes B_\mu$, where now $B_\mu$ is given by \eqref{eq:gauge_field}. This allows us to rewrite the fluctuated Dirac operator in the form
\begin{align*}
D_A = \sD\otimes\I + \gamma^\mu\otimes B_\mu + \gamma_5\otimes D_F = -i\gamma^\mu\otimes\nabla^E_\mu + \gamma_5\otimes D_F ,
\end{align*}
where we have defined a new connection $\nabla^E_\mu$ by
\begin{align}
\label{eq:conn_E}
\nabla^E_\mu = \nabla^S_\mu\otimes\I + i \I\otimes B_\mu .
\end{align}
For the square of the fluctuated Dirac operator, we obtain by direct calculation that 
\begin{align*}
{D_A}^2 = \Delta^E - Q ,
\end{align*}
where $\Delta^E$ is the Laplacian of the connection $\nabla^E$, and where $Q\in\Gamma(\End(M\times\mH_F))$ is given by
\begin{align*}
Q= -\frac14 s\otimes\I - \I\otimes {D_F}^2 + \frac12 i \gamma^\mu\gamma^\nu\otimes F_{\mu\nu} .
\end{align*}
Here we have defined the curvature $F_{\mu\nu}$ of the field $B_\mu$ as $F_{\mu\nu} := \partial_\mu B_\nu - \partial_\nu B_\mu$.

\begin{prop}
\label{prop:spec_act_ED}
The spectral action of the almost commutative manifold
\begin{align*}
M\times F\Sub{ED} = \left( C^\infty(M,\C^2), L^2(M,S)\otimes\C^4, \sD\otimes\I + \gamma_5\otimes D_F, \gamma_5\otimes\gamma_F, J_M\otimes J_F \right) 
\end{align*}
is given by
\begin{align*}
\Tr \left(f\Big(\frac {D_A}\Lambda\Big)\right) &\sim \frac1{4\pi^2} \int_M \L(g_{\mu\nu}, Y_\mu) \sqrt{|g|} d^4x + O(\Lambda^{-1}) ,
\end{align*}
for
\begin{align*}
\L(g_{\mu\nu}, Y_\mu) := 4\L_M(g_{\mu\nu}) + \L_Y(Y_\mu) + \L_H(g_{\mu\nu},d) .
\end{align*}
Here $\L_M(g_{\mu\nu})$ is defined in \cref{ex:canon_spec_act}. $\L_Y$ gives the kinetic term of the $U(1)$-gauge field $Y_\mu$ and equals 
\begin{align*}
\L_Y(Y_\mu) := \frac23f(0) \F_{\mu\nu}\F^{\mu\nu} ,
\end{align*}
where we have defined the curvature $\F_{\mu\nu}$ of the field $Y_\mu$ as $\F_{\mu\nu} := \partial_\mu Y_\nu - \partial_\nu Y_\mu$. 
The Higgs potential $\L_H$ (ignoring the boundary term) only gives two constant terms which add to the cosmological constant, plus an extra contribution to the Einstein-Hilbert action:
\begin{align*}
\L_H(g_{\mu\nu}) := -8f_2\Lambda^2 |d|^2 + 2 f(0) |d|^4 + \frac{1}{3} f(0) s|d|^2 .
\end{align*}
\end{prop}
\begin{proof}
Since ${D_A}^2$ is of the form $\Delta^E - Q$, we obtain the heat expansion of the spectral action from \eqref{eq:action_expansion}. Thus, we only need to calculate the Seeley-DeWitt coefficients from \cref{thm:seeley-dewitt}. The trace over the Hilbert space $\mH_F$ yields an overall factor $4 = \dim \mH_F$, so we obtain
\begin{align*}
a_0(x,{D_A}^2) = 4 a_0(x,\sD^2) .
\end{align*}
For the second coefficient we have
\begin{align*}
a_2(x,{D_A}^2) = 4 a_2(x,\sD^2) + \frac{1}{16\pi^2} \Tr\Big(-\I\otimes {D_F}^2 + \frac12i\gamma^\mu\gamma^\nu\otimes F_{\mu\nu}\Big) .
\end{align*}
Since $\Tr(\gamma^\mu\gamma^\nu) = 4g^{\mu\nu}$ and $F_{\mu\nu}$ is anti-symmetric, the trace over the last term vanishes. From \eqref{eq:Dirac} we easily see that ${D_F}^2 = |d|^2\I$, so we obtain
\begin{align*}
a_2(x,{D_A}^2) = 4 a_2(x,\sD^2) - \frac{|d|^2}{\pi^2} .
\end{align*}
For the last coefficient, we need the curvature $\Omega^E_{\mu\nu}$ of the connection $\nabla^E$ of \eqref{eq:conn_E}. Its square is given by
\begin{align*}
\Omega^E_{\mu\nu}{\Omega^E}^{\mu\nu} = \Omega^S_{\mu\nu}{\Omega^S}^{\mu\nu}\otimes\I - \I\otimes F_{\mu\nu}F^{\mu\nu} + 2i \Omega^S_{\mu\nu}\otimes F^{\mu\nu} ,
\end{align*}
where the last term is traceless. 
We also obtain a contribution from $Q^2$, which is given by
\begin{align*}
Q^2 = \frac{1}{16} s^2\otimes\I + \I\otimes {D_F}^4 - \frac14 \gamma^\mu\gamma^\nu\gamma^\rho\gamma^\sigma\otimes F_{\mu\nu}F_{\rho\sigma} + \frac12 s\otimes{D_F}^2 + \text{ traceless terms} .
\end{align*}
We shall ignore the boundary term $\Delta Q$. The last Seeley-DeWitt coefficient is then given by
\begin{multline*}
a_4(x,{D_A}^2) = 4 a_4(x,\sD^2) + \frac{1}{16\pi^2} \frac{1}{360} \Tr\Big(- 60 s\otimes {D_F}^2 + 180 \big(\I\otimes {D_F}^4 \\
- \frac14 \gamma^\mu\gamma^\nu\gamma^\rho\gamma^\sigma\otimes F_{\mu\nu}F_{\rho\sigma} + \frac12 s\otimes{D_F}^2\big) - 30 \I\otimes F_{\mu\nu}F^{\mu\nu} \Big) .
\end{multline*}
Using the trace identity
\begin{align*}
\Tr\big(\frac14 \gamma^\mu\gamma^\nu\gamma^\rho\gamma^\sigma\big) = g^{\mu\nu}g^{\rho\sigma} - g^{\mu\rho}g^{\nu\sigma} + g^{\mu\sigma}g^{\nu\rho} ,
\end{align*}
along with the anti-symmetry of $F_{\mu\nu}$, we calculate that 
\begin{align*}
\Tr\big(- \frac14 \gamma^\mu\gamma^\nu\gamma^\rho\gamma^\sigma\otimes F_{\mu\nu}F_{\rho\sigma}\big) = 2 \Tr(F_{\mu\nu}F^{\mu\nu}) = 8 \F_{\mu\nu}\F^{\mu\nu}.
\end{align*}
We thus obtain the final coefficient as
\begin{align*}
a_4(x,{D_A}^2) = 4 a_4(x,\sD^2) + \frac{1}{12\pi^2} s|d|^2 + \frac{1}{2\pi^2} |d|^4 + \frac{1}{6\pi^2} \F_{\mu\nu}\F^{\mu\nu} .
\end{align*}
The result now follows from inserting these Seeley-DeWitt coefficients into the asymptotic expansion \eqref{eq:action_expansion}, where we realize that the coefficients $a_k(\sD^2)$ yield the gravitational Lagrangian $\L_M$ of \cref{ex:canon_spec_act}. 
\end{proof}

\subsubsection{The fermionic action}

We have written the set of basis vectors of $\mH_F$ as $\left\{e_L, e_R, \bar{e_L}, \bar{e_R}\right\}$, and the subspaces $\mH_F^+$ and $\mH_F^-$ are spanned by $\left\{e_L, \bar{e_R}\right\}$ and $\left\{e_R, \bar{e_L}\right\}$, respectively. The total Hilbert space $\mH$ is given by $L^2(M,S)\otimes\mH_F$. Since we can also decompose $L^2(M,S) = L^2(M,S)^+ \oplus L^2(M,S)^-$ by means of $\gamma_5$, we obtain
\begin{align*}
\mH^+ = L^2(M,S)^+\otimes\mH_F^+ \oplus L^2(M,S)^-\otimes\mH_F^- .
\end{align*}
A spinor $\psi\in L^2(M,S)$ can be decomposed as $\psi = \psi_L + \psi_R$. Each subspace $\mH_F^\pm$ is now spanned by two basis vectors. A generic element of the tensor product of two spaces consists of sums of tensor products, so an arbitrary vector $\xi\in\mH^+$ can uniquely be written as 
\begin{align}
\label{eq:right_fermion_ED}
\xi = \chi_L\otimes e_L + \chi_R\otimes e_R + \psi_R\otimes\bar{e_L} + \psi_L\otimes\bar{e_R} ,
\end{align}
for Weyl spinors $\chi_L,\psi_L\in L^2(M,S)^+$ and $\chi_R,\psi_R\in L^2(M,S)^-$. Note that this vector $\xi\in\mH^+$ is now completely determined by two Dirac spinors $\chi := \chi_L + \chi_R$ and $\psi := \psi_L + \psi_R$. 

\begin{prop}
\label{prop:fermion_act_ED}
The fermionic action of the almost commutative manifold
\begin{align*}
M\times F\Sub{ED} = \left( C^\infty(M,\C^2), L^2(M,S)\otimes\C^4, \sD\otimes\I + \gamma_5\otimes D_F, \gamma_5\otimes\gamma_F, J_M\otimes J_F \right) 
\end{align*}
is given by
\begin{align*}
S_f = -i\big\langle J_M\til\chi,\gamma^\mu(\nabla^S_\mu - i Y_\mu)\til\psi\big\rangle + \langle J_M\til\chi_L,\bar d\til\psi_L\rangle - \langle J_M\til\chi_R,d\til\psi_R\rangle .
\end{align*}
\end{prop}
\begin{proof}
The fluctuated Dirac operator is given by 
\begin{align*}
D_A = \sD\otimes\I + \gamma^\mu\otimes B_\mu + \gamma_5\otimes D_F .
\end{align*}
An arbitrary $\xi\in\mH^+$ has the form of \eqref{eq:right_fermion_ED}, and then we obtain the following expressions:
\begin{align*}
J\xi &= J_M\chi_L\otimes\bar{e_L} + J_M\chi_R\otimes\bar{e_R} + J_M\psi_R\otimes e_L + J_M\psi_L\otimes e_R, \\
(\sD\otimes\I)\xi &= \sD\chi_L\otimes e_L + \sD\chi_R\otimes e_R + \sD\psi_R\otimes\bar{e_L} + \sD\psi_L\otimes\bar{e_R} , \\
(\gamma^\mu\otimes B_\mu)\xi &= \gamma^\mu\chi_L\otimes Y_\mu e_L + \gamma^\mu\chi_R\otimes Y_\mu e_R - \gamma^\mu\psi_R\otimes Y_\mu\bar{e_L} - \gamma^\mu\psi_L\otimes Y_\mu\bar{e_R} , \\
(\gamma_5\otimes D_F)\xi &= \gamma_5\chi_L\otimes \bar de_R + \gamma_5\chi_R\otimes d e_L + \gamma_5\psi_R\otimes d \bar{e_R} + \gamma_5\psi_L\otimes \bar d\bar{e_L} .
\end{align*}
We decompose the fermionic action into the three terms
\begin{align*}
\frac12 \langle J\til\xi,D_A\til\xi\rangle &= \frac12 \langle J\til\xi,(\sD\otimes\I)\til\xi\rangle + \frac12 \langle J\til\xi,(\gamma^\mu\otimes B_\mu)\til\xi\rangle + \frac12 \langle J\til\xi,(\gamma_5\otimes D_F)\til\xi\rangle ,
\end{align*}
and then continue to calculate each term separately. The first term is given by 
\begin{align*}
\frac12 \langle J\til\xi,(\sD\otimes\I)\til\xi\rangle = \frac12 \langle J_M\til\chi_L,\sD\til\psi_R\rangle + \frac12 \langle J_M\til\chi_R,\sD\til\psi_L\rangle + \frac12 \langle J_M\til\psi_R,\sD\til\chi_L\rangle + \frac12 \langle J_M\til\psi_L,\sD\til\chi_R\rangle .
\end{align*}
Using the fact that $\sD$ changes the chirality of a Weyl spinor, and that the subspaces $L^2(M,S)^+$ and $L^2(M,S)^-$ are orthogonal, we can rewrite this term as
\begin{align*}
\frac12 \langle J\til\xi,(\sD\otimes\I)\til\xi\rangle = \frac12 \langle J_M\til\chi,\sD\til\psi\rangle + \frac12 \langle J_M\til\psi,\sD\til\chi\rangle .
\end{align*}
Using the symmetry of the form $\langle J_M\til\chi,\sD\til\psi\rangle$, we obtain
\begin{align*}
\frac12 \langle J\til\xi,(\sD\otimes\I)\til\xi\rangle = \langle J_M\til\chi,\sD\til\psi\rangle = -i \langle J_M\til\chi,\gamma^\mu\nabla^S_\mu\til\psi\rangle .
\end{align*}
Note that the factor $\frac12$ has now disappeared in the result, and this is the reason why this factor is included in the definition of the fermionic action. The second term is given by
\begin{multline*}
\frac12 \langle J\til\xi,(\gamma^\mu\otimes B_\mu)\til\xi\rangle = - \frac12 \langle J_M\til\chi_L,\gamma^\mu Y_\mu\til\psi_R\rangle - \frac12 \langle J_M\til\chi_R,\gamma^\mu Y_\mu\til\psi_L\rangle \\
+ \frac12 \langle J_M\til\psi_R,\gamma^\mu Y_\mu\til\chi_L\rangle + \frac12 \langle J_M\til\psi_L,\gamma^\mu Y_\mu\til\chi_R\rangle .
\end{multline*}
In a similar manner as above, we obtain
\begin{align*}
\frac12 \langle J\til\xi,(\gamma^\mu\otimes B_\mu)\til\xi\rangle = -\langle J_M\til\chi,\gamma^\mu Y_\mu\til\psi\rangle ,
\end{align*}
where we have used that the form $\langle J_M\til\chi,\gamma^\mu Y_\mu\til\psi\rangle$ is anti-symmetric. The third term is given by
\begin{multline*}
\frac12 \langle J\til\xi,(\gamma_5\otimes D_F)\til\xi\rangle = \frac12 \langle J_M\til\chi_L,\bar d\gamma_5\til\psi_L\rangle + \frac12 \langle J_M\til\chi_R,d\gamma_5\til\psi_R\rangle \\
+ \frac12 \langle J_M\til\psi_R,d\gamma_5\til\chi_R\rangle + \frac12 \langle J_M\til\psi_L,\bar d\gamma_5\til\chi_L\rangle .
\end{multline*}
The bilinear form $\langle J_M\til\chi,\gamma_5\til\psi\rangle$ is again symmetric, but we now have the extra complication that two terms contain the parameter $d$, while the other two terms contain $\bar{d}$. Therefore we are left with two distinct terms:
\begin{equation*}
\frac12 \langle J\til\xi,(\gamma_5\otimes D_F)\til\xi\rangle = \langle J_M\til\chi_L,\bar d\til\psi_L\rangle - \langle J_M\til\chi_R,d\til\psi_R\rangle . \qedhere
\end{equation*}
\end{proof}

\begin{remark}
\label{remark:real_mass}
It is interesting to note that the fermions acquire mass terms without being coupled to a Higgs field. However, it seems we obtain a complex mass parameter $d$, where we would desire a real parameter $m$. By simply requiring that our result should be similar to \eqref{eq:coleman}, we will choose $d:=-im$, so that
\begin{align*}
\langle J_M\til\chi_L,\bar d\til\psi_L\rangle - \langle J_M\til\chi_R,d\til\psi_R\rangle = i \big\langle J_M\til\chi,m\til\psi\big\rangle .
\end{align*}
\end{remark}

The results obtained in this section can now be summarized into the following theorem.

\begin{thm}
\label{thm:ED}
The full Lagrangian of the almost commutative manifold
\begin{align*}
M\times F\Sub{ED} = \left( C^\infty(M,\C^2), L^2(M,S)\otimes\C^4, \sD\otimes\I + \gamma_5\otimes D_F, \gamma_5\otimes\gamma_F, J_M\otimes J_F \right) 
\end{align*}
as defined in this section, can be written as the sum of a purely gravitational Lagrangian,
\begin{align*}
\L\Sub{\text{grav}}(g_{\mu\nu}) = \frac1{\pi^2} \L_M(g_{\mu\nu}) + \frac1{4\pi^2} \L_H(g_{\mu\nu}) ,
\end{align*}
and a Lagrangian for electrodynamics,
\begin{align*}
\L\Sub{\text{ED}} = -i\Big( J_M\til\chi,(\gamma^\mu(\nabla^S_\mu - i Y_\mu)-m)\til\psi\Big) + \frac1{6\pi^2} f(0) \F_{\mu\nu}\F^{\mu\nu} .
\end{align*}
\end{thm}
\begin{proof}
The spectral action $S_b$ and the fermionic action $S_f$ are given by \cref{prop:spec_act_ED,prop:fermion_act_ED}. We shall now absorb all numerical constants into the Lagrangians as well. This immediately yields $\L\Sub{\text{grav}}$. To obtain $\L\Sub{\text{ED}}$, we need to rewrite the fermionic action $S_f$ as the integral over a Lagrangian. The inner product $\langle\;,\;\rangle$ on the Hilbert space $L^2(S)$ is given by 
\begin{align*}
\langle\xi,\psi\rangle = \int_M (\xi,\psi) \sqrt{|g|} d^4x ,
\end{align*}
where the hermitian pairing $(\;,\;)$ is given by the pointwise inner product on the fibres. Choosing $d=-im$ as in \cref{remark:real_mass}, we can then rewrite the fermionic action into
\begin{align*}
S_f &= -\int_M i\Big( J_M\til\chi,\big(\gamma^\mu(\nabla^S_\mu - i Y_\mu)-m\big)\til\psi\Big) \sqrt{|g|} d^4x . \qedhere
\end{align*}
\end{proof}

\section*{Acknowledgements}
\noindent We thank Jord Boeijink and Thijs van den Broek for valuable suggestions and remarks.

%\section{Conclusion}
%\bibliography{./Biblio_Thesis,references}

\newcommand{\noopsort}[1]{}\def\cprime{$'$}

\end{document}